\documentclass[10pt]{article}
\usepackage{dcolumn}
\usepackage{bm}
\usepackage{verbatim}       

\usepackage{amsthm}
\usepackage{amssymb}

\usepackage{amstext}
\usepackage{amsmath}
\usepackage{amsfonts}
\usepackage{units}
\usepackage{mathrsfs}
\usepackage{cite}
\newcommand{\p}{\partial}

\newcommand{\dd}{{\rm d}}

\newcommand{\bd}{\begin{definition}}                
\newcommand{\ed}{\end{definition}}                  
\newcommand{\bc}{\begin{corollary}}                 
\newcommand{\ec}{\end{corollary}}                   
\newcommand{\bl}{\begin{lemma}}                     
\newcommand{\el}{\end{lemma}}                       
\newcommand{\bp}{\begin{proposition}}            
\newcommand{\ep}{\end{proposition}}                
\newcommand{\bere}{\begin{remark}}                  
\newcommand{\ere}{\end{remark}}                     

\newcommand{\bt}{\begin{theorem}}
\newcommand{\et}{\end{theorem}}

\newcommand{\be}{\begin{equation}}
\newcommand{\ee}{\end{equation}}

\newcommand{\bit}{\begin{itemize}}
\newcommand{\eit}{\end{itemize}}
\newtheorem{theorem}{Theorem}[section]
\newtheorem{corollary}[theorem]{Corollary}
\newtheorem{lemma}[theorem]{Lemma}
\newtheorem{proposition}[theorem]{Proposition}
\theoremstyle{definition}
\newtheorem{definition}[theorem]{Definition}
\theoremstyle{remark}
\newtheorem{remark}[theorem]{Remark}

\begin{document}

\title{Completeness of Cauchy horizon generators}


\author{E. Minguzzi\thanks{
Dipartimento di Matematica e Informatica ``U. Dini'', Universit\`a
degli Studi di Firenze, Via S. Marta 3,  I-50139 Firenze, Italy.
E-mail: ettore.minguzzi@unifi.it} }

\date{}

\maketitle

\begin{abstract}
\noindent It is proved that every compactly generated future Cauchy horizon has past  complete generators, and dually. No condition on the differentiability of the horizon is imposed.
\end{abstract}

\section{Introduction}

The aim of this work is to prove that every compactly generated past Cauchy horizon $H^-(S)$ of a partial Cauchy hypersurface $S$ must have future geodesically complete generators, all that without imposing differentiability assumptions on the horizon which, being achronal, must already be locally Lipschitz \cite{lerner72}. It must be emphasized that placing strong differentiability conditions on the metric or on the partial Cauchy hypersurface $S$ does not guarantee that the horizon will be $C^1$ \cite{budzynski03}. Therefore, it is  necessary to study these objects without relying on too strong differentiability assumptions.


In order to formalize our problem let us recall some definitions.
A {\em spacetime} $(M,g)$ is a paracompact, time oriented Lorentzian manifold of dimension $n+1\ge 2$. We  assume that  $M$ is at least $C^{3}$ and that $g$ is $C^{2}$.
The {\em chronology violating set} is defined by $\mathcal{C}=\{p: p\ll p\}$, namely it is the (open) subset of $M$ made by those events through which passes a closed timelike curve.
A {\em lightlike line} is an achronal inextendible causal curve, hence a lightlike geodesic without conjugate points.  A  future inextendible causal curve $\gamma: [a,b)\to M$ is {\em totally future imprisoned} (or simply future imprisoned) in a compact set $K$, if there is $t_0\in [a,b)$ such that for $t>t_0$, $\gamma(t)\in K$.
A {\em partial Cauchy hypersurface} is an acausal edgeless (and hence closed) set.
Observe that for a partial Cauchy hypersurface $\textrm{edge} (H^{-}(S))=\textrm{edge} (S)=\emptyset$ (cf.\ \cite[Prop.\ 6.5.2]{hawking73}). Since every generator terminates at the edge of the horizon, the generators of $H^-(S)$  are future inextendible lightlike geodesic.
%


A past Cauchy horizon $H^{-}(S)$ is {\em
compactly generated} if there is a compact set $K$ such that every generator $\gamma$ of $H^{-}(S)$ is future
imprisoned in $K$.
Actually, we shall  use only the fact that each future generator $\gamma$ enters some compact set $K_\gamma$ that might depend on the generator. The notion of compactly generated Cauchy horizon has been introduced in Hawking's paper on chronology protection \cite{hawking92} although similar conditions appeared before \cite{tipler76}.

Our aim  is to prove the following result.

\begin{theorem} \label{one}
Let  $H^{-}(S)$ be a compactly generated past Cauchy horizon  for some  partial Cauchy hypersurface $S$. Then every lightlike generator $\gamma$ of $H^{-}(S)$ is
future complete (and analogously in the time-orientation-reversed case).
\end{theorem}

\begin{remark}
The acausality of $S$ in this  theorem is essential,
otherwise one gets a counterexample (to the past version) with $S$
the Misner boundary in the Misner spacetime \cite{stahl97}.
\end{remark}

This lemma appeared  in Hawking and Ellis' book  \cite[Lemma 8.5.5]{hawking73} and  was used several times\cite{galloway83,hawking92,beem98,krolak93,krasnikov94,budzynski01,rudnicki02} although no other proof can be found in  the literature.
As pointed out by Tipler \cite{tipler76,tipler77}  there are some (amendable) algebraic errors in Hawking and Ellis' proof.

There is also a  flaw at the conclusion of the proof where it is assumed without justification that the strong causality condition holds at the horizon (actually examples show that strongly causality does not need to hold at the horizon).  Budzy{\'n}ski, Kondracki and  Kr{\'o}lak \cite[Lemma 15]{budzynski01} recognized this problem and obtained the contradiction required by the proof changing just the last steps of the argument.

The main difficulty seems to be related with some tacit assumptions on the differentiability of the horizon. These difficulties have not been mentioned by previous authors probably because they had to assume the differentiability of the horizon anyway \cite{moncrief83,budzynski01,friedrich99}, for instance in order to apply   Hawking's area theorem \cite{hawking73}. Since  now we have at our disposal a  version of the area theorem for non-differentiable horizons \cite{chrusciel01}, it has become important to complement that result with Theorem \ref{one}.

Hawking and Ellis' proof of the completeness of the generators  requires the existence of a normalized $C^2$ geodesic vector field $V$ transverse to $H^{-}(S)$. If the horizon is $C^2$ this vector field can be constructed through the exponential map based on $H^{-}(S)$, however if $H^{-}(S)$ is not assumed $C^2$, and hence it is just Lipschitz, it is not at all clear that such geodesic  transverse field exists, particularly because the proof demands its existence in a whole neighborhood of $H^{-}(S)$, which seems a strong constraint especially for compact horizons.

One could try to approximate $H^{-}(S)$ with smooth hypersurfaces and then to use the exponential map from those, but it is unclear whether such an approach could be really successful since the Lipschitzness of $H^{-}(S)$ would imply an uncontrolled focusing of the geodesics normal to such approximating hypersurfaces. Even if the local problem could be solved one would face the issue of globalization. Observe that the sum of two geodesic timelike vector field is timelike but not necessarily geodesic.
Typically transversality conditions are quite demanding when it comes to smoothness issues. For instance,   the transversality of the boundary of a domain  with a continuous vector field reduces the family of Lipschitz domains to a subfamily (i.e.\ that of strongly Lipschitz domains \cite{hofmann07}) thus similar conditions cannot be assumed without justification.

This author does not know whether Hawking and Ellis' geodesic transversality assumption can be ultimately justified.
Nevertheless, since the result on the completeness of generators is crucial for some other results in mathematical relativity, we believe that a proof free from this problem should be given. We now proceed to give this proof. It removes altogether the geodesic assumption on $V$ at the price of introducing some analytical complications.

\subsection{Mathematical preliminaries:  Chaplygin's theorem}
We shall need the following known lemma. We include the proof for
completeness.

\begin{lemma} \label{pdg}
Let $F: [t_0,+\infty) \to \mathbb{R}$ be a $C^2$ function such that
$F(t) \to 0$ as $t \to +\infty$, and there is a constant $B>0$ such
that $\vert F''\vert \le B \vert F'\vert$ for $t>t_0$, then $F',F''
\to 0$ in the same limit.
\end{lemma}

\begin{proof}
 For every $\epsilon>0$ there is $t(\epsilon)>t_0$ such
that for $t>t(\epsilon)$, $\vert F\vert <\epsilon/4B$. Let $h>0$, by
Taylor expanding $F(t+2h)$ at $t>t(\epsilon)$ with a Lagrange
remainder
\[
F'(t)=\frac{1}{2h}[F(t+2h)-F(t)]-hF''(\zeta)  \qquad \zeta \in
[t,t+2h],
\]
we have
\[
\vert F'(t)\vert\le \frac{\sup_{t'>t(\epsilon)}\vert
F(t')\vert}{h}+h \sup_{t'>t(\epsilon)}\vert F''(t')\vert, \qquad
\forall h>0
\]
thus
\[
\sup_{t>t(\epsilon)} \vert F'(t)\vert\le \frac{\sup_{t>t(\epsilon)}
\vert F(t)\vert}{h}+h \sup_{t>t(\epsilon)}\vert F''(t)\vert, \qquad
\forall h>0
\]
The best choice of $h$ gives
\[
\sup_{t>t(\epsilon)} \vert F'(t)\vert\le 2
\sqrt{\sup_{t>t(\epsilon)}\vert F(t)\vert \,\sup_{t>t(\epsilon)} \vert F''(t)\vert }\le 2
\sqrt{B \sup_{t>t(\epsilon)} \vert F(t) \vert \sup_{t>t(\epsilon)}
\vert F'(t)\vert}
\]
thus  $\sup_{t>t(\epsilon)} \vert F'(t)\vert\le 4B
\sup_{t>t(\epsilon)} \vert F(t) \vert< \epsilon$, that is $F' \to 0$
and hence $F'' \to 0$ as $t \to +\infty$.
\end{proof}

The following result extends the inequality in Chaplygin's theorem (1919) \cite{chaplygin19,bertolino70,mitrinovic91} to the higher derivatives and provides also a quite simple proof of the stardard result (see also \cite{azbelev58}; a different, possibly more involved proof could pass through My{\u s}kis' induction  \cite{bertolino70}).

\begin{theorem} \label{aph}
Let $y: [0,c] \to \mathbb{R}$, be any $C^n$ function which  satisfies the differential inequality
\begin{equation} \label{vsv}
y^{(n)}+a_1(x) y^{(n-1)}+\cdots+ a_{n}(x) y\ge 0,
\end{equation}
and the initial conditions
\[
y^{(i)}=0, \textrm{ for } i=0,1,\cdots,n-1,
\]
where the coefficients $\{a_i(x)\}$ are Lipschitz.
Then there is some constant $b$,  $0<b\le c$, dependent on the coefficients $\{a_i(x)\}$ but independent of $y$, such that for every $y$  and every $i=0,\cdots, n-1$ we have  $y^{(i)}\ge 0$ on $[0,b]$. Moreover, the equality $y=0$ holds on the whole interval $[0,b]$ if and only if the equality in   (\ref{vsv}) holds  on the whole interval $[0,b]$.
\end{theorem}

\begin{proof}
Let us denote with $L$ the linear differential operator on the left-hand side of (\ref{vsv}). Let $K(x,\xi)$ be the Cauchy function, that is the function defined on the triangle $0\le \xi \le x < c/2$  which solves $L K=0$ for every $\xi \in[0,c/2)$ with initial condition
\[
K(\xi,\xi)=K^{(1)}(\xi,\xi)=\cdots=K^{(n-2)}(\xi,\xi)=0, \quad  K^{(n-1)}(\xi,\xi)=1,
\]
where these derivatives are with respect to $x$. This function is well defined because it is equivalently determined by the linear system of first-order differential equations in the dependent variables $(K^\xi_0, K^\xi_1, K^\xi_2,\cdots, K^\xi_{n-1}) : [0,c/2]\to \mathbb{R}$
\begin{align*}
\frac{d}{ds}{K^\xi}_0&=K^\xi_1, \\
\cdots &\\
\frac{d}{ds}{K^\xi}_{n-2}&= K^\xi_{n-1},\\
\frac{d}{ds} {K^\xi}_{n-1}&=- a_1(s+\xi) K^\xi_{n-1}-a_2(s+\xi) K^\xi_{n-2}\cdots-a_n(s+\xi) K^\xi_0 .
\end{align*}
where $s+\xi=x$, under the initial condition
\[
K^\xi_0(0)=K^\xi_1(0)=\cdots=K^\xi_{n-2}(0)=0, \quad  K^\xi_{n-1}(0)=1.
\]
Indeed, there is one and only one solution by the Picard-Lindel\"of theorem \cite[Cor.\ 5.1]{hartman64}, and so we obtain our desired Cauchy function once we set $K(x,\xi)=K^\xi_0(x-\xi)$ (thus $K^{(i)}(x,\xi)=K^\xi_i(x-\xi)$). Moreover, the above system of first-order differential equations is Lipschitz also with respect to the external parameter $\xi$, thus its solution $(K^\xi_0, K^\xi_1, K^\xi_2,\cdots, K^\xi_{n-1}) $ has a Lipschitz dependence on $(s,\xi)$ (this is Peano's theorem; if the $a$s are $C^1$ then one can infer that the $K$s are $C^1$ too  \cite[Theor.\ 3.1]{hartman64}, for the Lipschitz case see \cite{lang95,cartan71}). The function $K^{(n-1)}(x,\xi)$ is continuous in both $(x,\xi)$ on the triangle $0\le \xi \le x < c/2$ and in particular at $(0,0)$. Since $K^{(n-1)}(0,0)=1>0$,  there is some neighborhood of $(0,0)$ (in the product topology) over which  $K^{(n-1)}$ is positive, and hence a triangle $0\le \xi \le x < b\le c/2$ over which  $K^{(n-1)}$ is positive. But on the diagonal $K^{(i)}$, $i=1,\cdots, n-2$, vanishes so upon integration on $x$ we obtain that $K^{(i)}$, is positive on $0\le \xi < x < b$ for every $i=1,\cdots, n-1$.

From the assumption we have that $Ly\ge 0$, where $Ly$ is continuous. The  uniqueness of the solution to the differential equation $L y= f$  implies the easily verifiable   formula
\[
y(x)=\int_0^x K(x,\xi) Ly(\xi) \, \dd \xi ,
\]
 which under differentiation gives  more generally
\[
y^{(i)}(x)=\int_0^x K^{(i)}(x,\xi) Ly(\xi)\, \dd \xi , \qquad i=0,1,\cdots, n-1,
\]
thus $y^{(i)}\ge 0$ on $[0,b]$, and the equality $y=0$ on $[0,b]$ is possible only if $Ly=0$ on $[0,b]$.
\end{proof}

\section{Geodesic incompleteness and timelike variations}

Let $\eta$ be a future inextendible causal curve. We recall  that  $\Omega_f(\eta)$ denotes the
 closed set of
accumulation points to the future of $\eta$, that is,
\[
\Omega_f(\eta)=\bigcap_t \overline{\eta([t,+\infty))},
\]
and
analogously in the past case. The set $\Omega_f(\eta)$ is
non-empty and compact if and only if $\eta$ is totally future
imprisoned \cite[Prop.
3.2]{minguzzi07f}. Furthermore, in this  case $\Omega_f(\eta)$ is connected, it is
the intersection of all the compact sets in which $\eta$ is future
imprisoned and
for every relatively compact open set
$U\supset \Omega_f(\eta)$, $\eta$ is future imprisoned in $\bar{U}$.

In \cite[Prop. 3.2]{minguzzi07f} we have also established that  if a
future inextendible totally future imprisoned causal curve $\gamma$
does not intersect the chronology violating set, then there is a
non-empty achronal compact set - the {\em minimal invariant set}
$\Omega\subset \Omega_f(\gamma)$ - with many interesting properties:
it is generated by lightlike lines contained in the set, and given
one such line $\eta: \mathbb{R}\to M$,
$\bar{\eta}=\Omega_f(\eta)=\Omega_p(\eta)=\Omega$ (see also \cite{kay97}).

The proof of the next theorem solves the main technical problem of the paper. Unfortunately, it is rather long.

\begin{theorem} \label{oih}
Let $(M,g)$ be a spacetime and let $\gamma: [0,v_{\infty}) \to
M$, $v_{+\infty}<+\infty$, be a future inextendible future
incomplete lightlike geodesic, with affine parameter $v$, totally
future imprisoned in a compact set. Then there is a
future directed timelike variational field on $\gamma$ which goes to zero\footnote{This notion is well defined, for it  is independent of the metric with which we measure  this field, since $\gamma$ is future imprisoned in a compact set.} as $v
\to v_{\infty}$
and whose induced variation gives, for sufficiently small
variational parameter, a future inextendible totally future
imprisoned timelike curve $\sigma$ such that $\Omega_f(\sigma)=
\Omega_f(\gamma)$.
\end{theorem}


\begin{proof}
Since $M$ is time oriented we can find on it a future directed $C^2$ timelike vector field $V$ such that $g(V,V)=-1$.
Let $A$ be a relatively compact set which contains the compact set\footnote{Trying to redefine $V$ so as to impose $D_V V=0$ on a neighborhood of $C$ leads to various problems of differentiability, related to the fact that $C$ is a subset of an achronal and hence just Lipschitz hypersurface, and to the fact that the induced topology of $C$ cannot be recovered from the real line topology of its generators.} $C=\gamma\cup \Omega_f(\gamma)$.
Let us define on $A$ the Riemannian metric
\begin{equation} \label{nka}
g'(X,Y)=\frac{1}{2}g(X,Y)+g(X,V) g(Y,V),
\end{equation}
and let $t(v)$ be the parameter that measures the $g'$-length along
$\gamma$ and such that $t(0)=0$. Let $\tilde\gamma(t):=\gamma(v(t))$ and let a dot denote  differentiation with respect to $t$. Let $\p_t$ be the tangent vector to $\tilde\gamma$ and let $\p_v$ be the tangent vector to $\gamma$.
 From the definition of $t$,  $g(V,\p_t)=-1$. As $\gamma$ has no future endpoint $t$ has no upper
bound (the argument goes as in the last part of the proof of
\cite[Lemma 3.65]{beem96}, note that the completeness of $g'$ is not
needed because $\gamma$ is  imprisoned in a compact set). Let $\kappa(t)$
and $h(t)>0$ be given by
\[
D_t \p_t=\kappa \,\p_t , \qquad
\p_v=h \,\p_t, \qquad \textrm{where }
\kappa=-h^{-1}\dot{h},
\]
and where $D$ is the Levi-Civita covariant derivative compatible with $g$.
The $g'$-unit  subbundle of $T\bar{A}$ is compact. As a consequence,  the continuous
quantities
\[
\vert g(D_N V,N)\vert; \ \vert g(D_N V,D_N
V)\vert; \ \vert g(N, R(V,N)V)\vert;  \ \vert g(N, D_V V)\vert;  \ \vert g(N, D_N D_V V)\vert,
\]
with $N$ arbitrary $g'$-normalized vector, are bounded by a positive constant $K$ on
the compact set $\bar{A}$. Note that
\[
 \kappa=-\kappa \,g\left(V,\p_t\right)=-g\left(V,D_t\p_t \right)=g\left(D_t V, \p_t\right)
\]
thus as $\p/\p t$ is $g'$-normalized, $\vert \kappa\vert \le K$.
Since $\gamma$ is future incomplete the affine parameter
\[
v(t)=\int_0^{t}h^{-1}(s)\,\dd s\] has a finite limit $v_{\infty}$ as $t
\to +\infty$. Let us consider the function $F(t)=v(t)-v_{\infty}$. It is
$C^2$, with $F'=h^{-1}$, $F''=\kappa h^{-1}$, thus  $\vert F''\vert \le
K\vert F'\vert $. By lemma \ref{pdg} $F' \to 0$,
$F'' \to 0$, that is $h^{-1} \to 0$ and $\frac{\dd h^{-1}}{\dd t}
\to 0$ as $t \to +\infty$. The positive function
\begin{equation}
\label{pvg} x(t)=\frac{h^{-1}(t)}{2v_{\infty}-v(t)}, \quad t \in
[0,+\infty),
\end{equation}
is clearly such that both $\vert  x\vert $ and $\vert \dot{x}\vert$
are bounded by some constant $E> K$. From the definition of $x(t)$ we find
\[
\dot{x}=\kappa x+x^2,
\]
thus
\[
\vert \dot{x}\vert\le 2E x .
\]
For every $p\in C$ we can find some constant $c>0$ such that $\exp_p(sV)$ exist for every $s\in [0,c]$ and is contained in $A$. By continuity of the exponential map on the base point (recall that the geodesic equation is a first order differential equation on the tangent bundle; the continuous dependence on the initial condition is proved in \cite{hartman64}) we have that there is a neighborhood $O$ of $p$ such that for every $q\in O$, $\exp_q(sV)$ exist for every $s\in [0,c]$ and belongs to $A$.
Since $C$ is compact we can find a  constant $B>0$ such that for every $0\le U <B$,  and every $p\in C$, $\exp_p (UV)$ exists and belongs to $A$.


Now, consider the variation of $\gamma$ towards the future given by
\begin{equation} \label{vad}
\alpha(t,u)=\exp_{\tilde\gamma(t)}(x(t)V u)
\end{equation}
where $x$ is given by Eq. (\ref{pvg}) and $\frac{\p}{\p u}=xV$ is
the variational field. Observe that the variational field vanishes in the limit since $\vert x\vert \to 0$.
 Since $\vert x\vert$ is bounded, it
is  possible to find an $\epsilon>0$ such that the exponential map of Eq. (\ref{vad})
is well defined for $u<\epsilon$ and with value in $A$ for every $t$.

We want to establish whether the
varied curve is timelike for sufficiently small $u$.
Since $\p_t$ and $\p_u$ commute we have over $\gamma$ ($u=0$)
\begin{align*}
\frac{1}{2}\, \p_u\, g\left(\p_t,\p_t \right)\Big\vert_{u=0}&
=  \p_t \,g\left(\p_u,\p_t \right)-g\left(\p_u, D_t \p_t\right)\\
&=-  \left(\dot{x}+h^{-1}\dot{h}
x\right)=-x \frac{\dd \ln (hx)}{\dd t}=- x^2
\end{align*}
where we used Eq. (\ref{pvg}).  Thus at any fixed $t$, the varied curve is timelike provided we restrict $u$ to some interval  $0<u<\varepsilon(t)$. We want to show that we can find $\varepsilon$ independent of $t$, and hence that a timelike variation exists. This is done controlling the second derivative on a subset $\mathbb{R}^+\times [0,\epsilon)$ of the $(t,u)$ space.

Let us calculate the second derivative (possibly $u\ne 0$)
\[
\frac{1}{2} \p_u^2\, g\!\left(\p_t,\p_t\right)=
\p_u \,g\!\left(\p_t,D_t \p_u\right) = g\!\left(D_t
\p_u, D_t \p_u \right)+g\!\left(\p_t,D_t D_u
\p_u\right)+g\!\left(\p_t,R\!\left(\p_u,\p_t\right)\!\p_u\right).
\]
Since $V$ is $g$-normalized in $U$ and $x$ does not depend on $u$, the
right-hand side reads
\[
-\dot{x}^2+x^2\left[g\!\left(D_t V, D_t V\right)+ g\!\left(\p_t,R\!\left(V,\p_t\right)\!V\right)\right]+x^2 g(\p_t, D_t D_V V)+2x \dot{x}  g(\p_t, D_V V)
\]
for $0\le u<\epsilon$. Note that $\p_t$ is not necessarily
$g'$-normalized  for $u>0$, however $\p_t/\sqrt{g'(\p_t,\p_t)}$ is, thus we have
\[
\p^2_u\,
g\!\left(\p_t,\p_t\right)\le 8K x^2 \{g'\!
\left(\p_t,\p_t\right)+E  {g'\!
\left(\p_t,\p_t\right)}^{1/2} \}.
\]
Let us use the inequality, which holds for any positive $a,b,c$
\[
b+a\sqrt{b}\le c+(1+\frac{a^2}{4c}) b,
\]
so as to obtain with $b= g'\!
\left(\p_t,\p_t\right)$, $a=E $, $c=E^2/2$
\begin{equation} \label{nnj}
\p^2_u\,
g\!\left(\p_t,\p_t\right)\le 4K x^{2}\{E^2+ 3 g'\!
\left(\p_t,\p_t\right)\}.
\end{equation}

Let us leave this equation for the moment. Since $V$ is  $g$-normalized
\begin{align*}
\p_u g\!\left(V,\p_t \right)&=x g\!\left(D_V V,\p_t \right)+g\!\left(V,D_t \p_u\right)\\&=x g\!\left(D_V V,\p_t \right)+g\!\left(V,x D_t V+V \dot{x}
\right)\\
&=x g\!\left(D_V V,\p_t\right)-\dot{x}.
\end{align*}
 Moreover,  since for $u=0$,
$g\!\left(V,\p_t\right)=-1$, we have for any $t $ and $0\le u<B$
\begin{align*}
\vert g\!\left(V,\p_t\right) \vert&=\vert -1+\int_0^u [x g\!\left(D_V V,\p_t \right)-\dot{x}] \,\dd u'\vert\\
&\le x K \int_0^u   g'\!
\left(\p_t,\p_t \right)^{1/2}  \,\dd u'+1+EB .
\end{align*}
From Eq. (\ref{nka}) and the triangle inequality
\[
g'\!
\left(\p_t,\p_t \right)^{1/2}\le \vert g\!
\left(\p_t,\p_t \right)\vert ^{1/2}+\vert g\!\left(V,\p_t\right) \vert,
\]
which substituted into the previous equation gives
\begin{align*}
\vert g\!\left(V,\p_t\right) \vert&\le \{  x K \int_0^u  \vert g\!
\left(\p_t,\p_t \right)\vert^{1/2}  \,\dd u'+1+EB  \}+ x K \int_0^u  \vert g\!\left(V,\p_t\right) \vert  \,\dd u'\\
&\le  \{  x K \sqrt{u} [\int_0^u  \vert g\!
\left(\p_t,\p_t \right)\vert  \,\dd u']^{1/2}+1+EB  \}+ x K \int_0^u  \vert g\!\left(V,\p_t\right) \vert  \,\dd u' ,
\end{align*}
where we used the Cauchy-Schwarz inequality.
By the Gronwall inequality, defined $\Psi(u)=  x K \sqrt{u} [\int_0^u  \vert g\!
\left(\p_t,\p_t \right)\vert  \,\dd u']^{1/2}+1+EB$
\begin{align*}
\vert g\!\left(V,\p_t\right) \vert&\le\Psi+xK\int_0^u \Psi(s) e^{x K (u-s)} \dd s\le \Psi+xKe^{E K B}\int_0^u \Psi  \dd s\\
&\le (1+EB ) (1+e^{E K B}EKB)+ x K \sqrt{u}\, [\int_0^u  \vert g\!
\left(\p_t,\p_t \right)\vert  \,\dd u']^{1/2}\\
& \quad+ e^{E K B} x^2 K^2 \int_0^u   \sqrt{s} \,[\int_0^s  \vert g\!
\left(\p_t,\p_t \right)\vert  \,\dd u']^{1/2} \dd s \\
&\le (1+EB ) (1+e^{E K B}EKB)+ x K \sqrt{u} \,[\int_0^u  \vert g\!
\left(\p_t,\p_t \right)\vert  \,\dd u']^{1/2}\\
& \quad+ \frac{K^2}{\sqrt{2}} \,e^{E K B} x^2 u  [\int_0^u    \int_0^s  \vert g\!
\left(\p_t,\p_t \right)\vert  \,\dd u'\dd s ]^{1/2} ,
\end{align*}
where in the last step we used once again the  Cauchy-Schwarz inequality.

We have shown that there are positive constants $\chi,\psi,\omega$,
\begin{align*}
\chi &=(1+EB ) (1+e^{E K B}EKB); \quad \psi=K; \quad \omega=\frac{K^2}{\sqrt{2}} \,e^{E K B} ,
\end{align*}
such that defined
\[
y=x^2\int_0^u    \int_0^s  \vert g\!
\left(\p_t,\p_t \right)\vert  \,\dd u'\dd s ,
\]
Eq.\ (\ref{nka}) with the bound for $\vert g\!\left(V,\p_t\right) \vert$ just found implies
\[
g'\! \left(\p_t,\p_t \right)\le \frac{1}{x^2}\, \p^2_u y+(\chi+\psi  \sqrt{u}\, \sqrt{\p_u y}+\omega x u \sqrt{y})^2 \le \frac{1}{x^2}\, \p^2_u y+4(\chi^2+\psi^2  u \p_u y+\omega^2 x^2 u^2 y) .
\]
Substituting into the right-hand side of Eq.\ (\ref{nnj}), we obtain that in a maximal connected (non-empty) neighborhood of $u=0$, where $g\!\left(\p_t,\p_t\right)(u)$ is non-positive, the following inequality holds
\begin{align}
-\frac{1}{x^4}\,\p^4_u\,
y& \le  4K\{E^2+12 \chi^2+\frac{3}{x^2}\p^2_u y+12\psi^2  u \p_u y+12\omega^2 x^2 u^2 y\}  ,  \label{inf}
\end{align}
where at $u=0$, the function $y(u)$ satisfies
\[
y(0)=0, \quad \p_u y(0)=0, \quad \p^2_u y(0)=0, \quad \frac{1}{x^3}\, \p^3_u y(0)=2 x.
\]
Let us consider a function $w$ which satisfies
\begin{align}
-\frac{1}{x^4}\, \p^4_u\,
w& =  4K\{E^2+12 \chi^2+\frac{3}{x^2}\p^2_u w+12\psi^2  u \p_u w+12\omega^2 x^2 u^2 w\}, \label{ing}
\end{align}
with the initial conditions
\[
 w(0)=0, \quad \p_u  w(0)=0, \quad \p^2_u  w(0)=0,  \quad \frac{1}{x^3}\,   \p^3_u  w(0)=0,
\]
and let us consider a function $z$ which satisfies
\begin{align}
-\frac{1}{x^4}\, \p^4_u\,
z& =  4K\{\frac{3}{x^2}\p^2_u z+12\psi^2  u \p_u z+12\omega^2x^2 u^2 z\}, \label{inh}
\end{align}
with the initial conditions
\[
 z(0)=0, \quad \p_u  z(0)=0, \quad \p^2_u  z(0)=0, \quad \frac{1}{x^3}\,  \p^3_u z(0)=2 x.
\]
These functions exist and are unique because the right-hand sides of these differential equations are Lipschitz in the variables $(w,\p_u w,\p_u^2 w, \p^3_uw, u)$ or in the variables $(z,\p_u z,\p_u^2 z, \p^3_u z, u)$. Existence and uniqueness of $C^4$ solutions in a neighborhood of $u=0$ follows then from the   Picard-Lindel\"of theorem once they are rewritten as a system of first order differential equations through a standard trick. We observe that $w+z$ satisfies Eq.\ (\ref{inf}) with the equality sign and has the same initial conditions of $y$.

These differential equations  depend on $t$ just through $x(t)$, and this dependence can be removed from the differential equation introducing the variable $U=xu$. Let $x_m=x(t_m)$ be the maximum value of $x$ (recall that $x\to 0$ for $t\to \infty$, so the maximum exists).
Equation (\ref{ing}) has initial conditions which do not depend on $t$ once expressed in this variable, thus there is a positive constant $b$ and some function $W:[0,x_m b)\to \mathbb{R}$ such that
\[
w(t,u)=W(x(t)u).
\]
From Eq.\ (\ref{ing})  we get the first terms of the Taylor expansion of $W''$
\[
W''(U)=-2K(E^2+12\chi^2) U^2+o_1(U^3).
\]
Observe that $w$ is defined for $u\in [0, \frac{x_m}{x(t)}\, b)$ and so for $u\in [0,b)$.
The differential equation (\ref{inh}) is linear homogeneous in $z$, thus if $Z(x_m u)$ is the solution for $t=t_m$, we have for arbitrary $t$
\[
z(t,u)=\frac{x(t)}{x_m} \,Z(x(t) u),
\]
where at $U=0$, $Z'''=2\,x_m$, $Z''''=0$, so that
\[
Z''(U)=2\,x_m U + o_2(U^2)
\]
Thus $b$ can be chosen so small that $x_m b<B$  and for $0<U<x_m b$,
\[
2x_m- 2K (E^2+12\chi^2) b> Eb\frac{b\vert o_1(U^3)\vert+\frac{1}{x_m}U \vert o_2(U^2)\vert }{U^3}
\]
For any $t$, let $u\in [0,b)$ then $U\in [0,x_m b)$ and
\begin{align*}
 \p^2_u (w+z)&=x^2[-2K(E^2+12\chi^2)U^2 +o_1(U^3)+ \frac{x}{x_m}\, 2\,x_m U+\frac{x}{x_m}\, o_2(U^2)]\\
 &\ge \frac{x^2}{b} [2-2K(E^2+12\chi^2)b] U^2-x^3\vert \frac{1}{x} o_1(U^3)+\frac{1}{x_m}\, o_2(U^2)\vert\\
 &\ge \frac{x^2}{b} [2-2K(E^2+12\chi^2)b] U^2-x^3\{ b \vert o_1(U^3)\vert /U+\frac{1}{x_m}\vert o_2(U^2)\vert\}\\
 &\ge \frac{x^2 U^2}{b} \{2-2K(E^2+12\chi^2)b- Eb \frac{ b \vert o_1(U^3)\vert +\frac{1}{x_m}\, U \vert o_2(U^2)\vert}{U^3} \}\ge  0 ,
\end{align*}
where the last inequality is strict for $U>0$.

The function $\tilde{y}=y-(w+z)$ regarded as a function of $U$, has the initial conditions $\p^i_U \tilde{y}=0$, $i=0,1,2,3$, and satisfies the differential inequality
\begin{align}
\frac{1}{12K}\,\p^4_U\,
\tilde y + \p^2_U \tilde y+4\psi^2  U \p_U \tilde y+4\omega^2 U^2 \tilde y  \ge 0,
\end{align}
thus by Theor.\ \ref{aph} there is some interval $[0,x_m \tilde{b})$, $\tilde{b}>0$, for the variable $U$ over which $\p^2_U \tilde{y} \ge 0$ hence $ \p^2_u y\ge \p^2_u (w+z)$ for $u\in [0,\tilde{b}]$. Thus $b>0$ can be chosen so small that $x^2 \vert g(\p_t,\p_t)\vert=\p^2_u y > 0$ on $(0,b]$ independently of the value of $t$.

We already known that for each $t$
there is an interval  $(0,\epsilon(t))$ such that for $u$ belonging to this interval $g(\p_t,\p_t)<0$. The just found result proves that at any $t$ we  can take $\epsilon(t)=b$, for otherwise $\p^2_u y$ would have to vanish for some $u\in (0,b)$ which is impossible.

The last statement on the equivalence $\Omega_f(\sigma)=
\Omega_f(\gamma)$ follows from $x\to 0$ and from the fact that the norm of  $V$, with respect to any auxiliary Riemannian metric, is bounded on $\bar{A}$.
\end{proof}

\subsection{Completeness of compactly generated Cauchy horizons}

In this subsection we apply the previous theorem to solve our  main problem and explore further consequences.
Indeed, we are ready to prove the main result of this work.

%

\begin{proof}[Proof of Theorem \ref{one}]
Let $p
\in H^{-}(S)$ and let us denote with $\gamma$ the portion of the generator passing through $p$ to the future of $p$. Since $S$ is edgeless, $\gamma$ is future inextendible. Moreover,
since $S$ is a partial Cauchy surface $H^-(S)\cap S=\emptyset$. As
$H^{-}(S)$ is closed,  $\Omega_f(\gamma) \subset H^{-}(S)$. Since
$\Omega_f(\gamma)$ is compact, $S$ is closed and
$\Omega_f(\gamma)\cap S=\emptyset$ we can find an open neighborhood
$U\supset \Omega_f(\gamma)$ such that $\bar{U}\cap S=\emptyset$,
$\bar{U}$ is compact, and $\gamma$ is contained in
$\bar{U}$.
Suppose that $\gamma$ is not future complete, and consider the future inextendible timelike curve $\sigma$
obtained through the variation of $\gamma$ to the future as in
Theorem \ref{oih} (we are not demanding that $\gamma\subset \Omega_f(\gamma)$). If the variational parameter is taken sufficiently
small the starting point of $\sigma$ belongs to $Int D^{-}(S)$. As
$\Omega_f(\sigma)=\Omega_f(\gamma)\subset U$ (or using the compactness of $\Omega_f(\gamma)\cup \gamma$), the variational parameter can be chosen so small that $\sigma$ is totally imprisoned
in $\bar{U}$. However, it must intersect the Cauchy surface $S$, so we obtain the desired
contradiction.
\end{proof}

%

If we know that a $C^0$ future null hypersurface \cite{galloway00}, not necessarily a Cauchy horizon, is generated by almost closed geodesic then we can obtain stronger results on the relationship between geodesic incompleteness of the generators and chronology violation.

\begin{theorem} \label{kso}
Let $\gamma$ be a future inextendible future incomplete lightlike
geodesic,  future imprisoned in a compact set and such that $\gamma
\subset \Omega_f(\gamma)$. Then for every $p \in \gamma$, and neighborhood $U\supset \Omega_f(\gamma)$, there is some closed timelike curve contained in $I^+_U(\gamma)$.
\end{theorem}

\begin{proof}
Let us cut a first segment of $\gamma$ so as to make $p$ its starting point.
Consider the future inextendible timelike curve $\sigma(t)=\exp_{\tilde\gamma(t)}(xVb) \subset I^+(\gamma)$
constructed in theorem \ref{oih} through a timelike variation to the
future of $\gamma$. Let us take the variational parameter $b$ so small that the starting point $q$ of $\sigma$ stays in $U$. Since $\Omega_f(\gamma)\supset \gamma$ is compact and $x$ is bounded we can also find $b$ so small that $\sigma \subset U$. Since
$\Omega_f(\sigma)=\Omega_{f}(\gamma)\supset \gamma \ni p$, $\sigma$ enters indefinitely
any neighborhood  of $p$, in particular it enters $I_U^{-}(q)$. Thus $\sigma\cap I^-_U(q)\ne \emptyset$ which allows us to construct a  closed timelike curve contained in $I^+_U(\gamma)$.
\end{proof}

The next result establishes that compact Cauchy horizons are rather special null hypersurfaces.

\begin{theorem}
Let  $S$ be a partial Cauchy hypersurface. If $H^{-}(S)$ is compactly
generated then  it contains a lightlike line $\eta$ such that $\overline{\eta} = \Omega_f(\eta)=\Omega_p(\eta)$ and either this line is complete or $\eta$ belongs to the boundary of the chronology violating set.
\end{theorem}

Here $\overline{\eta}$ is also a minimal invariant set, see \cite{minguzzi07f} for details.

\begin{proof}
Let $\gamma$ be a future inextendible generator of $H^-(S)$.  Clearly, $H^{-}(S)$ does not intersect the chronology violating set $\mathcal{C}$, indeed $S$ does not intersect it because it is acausal, and so no point $r$ of $H^{-}(S)$ can belong to $\mathcal{C}$ otherwise the closed timelike curve passing through $r$ would provide a future inextendible timelike curve not intersecting $S$.

The existence of $\eta$ is now a consequence of \cite[Theor.\ 3.9]{minguzzi07f} (see also \cite[Prop.\ 1]{kay97}). The curve $\eta$ is  future complete because every generator of a compactly generated past Cauchy horizon has this property. Moreover, it is past imprisoned in the compact set $\bar{\eta}$ thus it is  either past complete or it belongs to $\overline{\mathcal{C}}$ (see the proof of Theorem \ref{kso}).
\end{proof}

The existence of a complete lightlike line $\eta$ implies that one among  the null convergence condition and the null genericity condition do not hold on $\eta$. Typically the former property is assumed thus this fact establishes that the horizon at $\eta$ has rather special geometry.

\section{Conclusions}

We have established that for compact or compactly generated past Cauchy horizons which are not necessarily differentiable the classical theorem according to which the generators are future geodesically complete  still holds. This result and its dual are expected to be useful in the study of Cauchy horizons but also, given the broad applicability of Theorem \ref{oih}, in the study of  general  $C^0$  null hypersurfaces \cite{galloway00} which are compactly generated (i.e.\ whose null generators are  imprisoned in a compact set). This theorem is essential in order to infer that the expansion $\theta$ is non-negative almost everywhere over $H^{-}(S)$, a fact  used in many results and arguments of mathematical relativity.

\section*{Acknowledgments}  I thank Alexander Domoshnitsky for some useful comments on the history and development of Chaplygin's type theorems. This work has been
partially supported by GNFM of INDAM.


\def\cprime{$'$}

\end{document}